\documentclass[11pt, notitlepage]{article}
\usepackage[utf8]{inputenc}
\usepackage[T1]{fontenc}
\usepackage{graphicx} 	% Include figure files
\usepackage[compress, square, numbers]{natbib}
\usepackage{amsmath, amsthm, amssymb}
\usepackage{bm}			% bold math
\usepackage{mathtools}  	% extra mathtools, some things fixes AMS math bugs
\usepackage{dsfont}		% Double stroke roman fonts, e.g. \mathds{1}
\usepackage{centernot}  	% for centered negations (cross out slashes)
\usepackage{fullpage}
\usepackage{sidecap}
\usepackage{hyperref}
\usepackage{xcolor}
\hypersetup{
    colorlinks,
    linkcolor={red!50!black},
    citecolor={blue!50!black},
    urlcolor={blue!80!black}
}

\newtheorem{theorem}{Theorem}	
\newtheorem{proposition}{Proposition}
\newtheorem{corollary}{Corollary}

\newtheorem{conjecture}{Conjecture}

\theoremstyle{definition}
\newtheorem{definition}{Definition}

\theoremstyle{remark}
\newtheorem{example}{Example}

\DeclareMathOperator{\tr}{tr}

\DeclareMathOperator{\wt}{wt}
\DeclareMathOperator{\supp}{supp}

\newcommand{\C}{\mathbb{C}}

\newcommand{\DD}{\mathcal{D}}

\newcommand{\CC}{\mathcal{C}}
\newcommand{\GG}{\mathcal{G}}

\newcommand{\on}[1]{\operatorname{#1}}
\newcommand{\ot}[0]{\otimes}
\newcommand{\nn}[0]{\nonumber}

\newcommand{\nhf}[0]{\lfloor n/2 \rfloor}

\renewcommand{\P}[0]{\mathcal{P}}
\newcommand{\Pe}[0]{P_{\rm{e}}}
\newcommand{\Po}[0]{P_{\rm{o}}}
\newcommand{\Ae}[0]{A_{\rm{e}}}
\newcommand{\Ao}[0]{A_{\rm{o}}}

\newcommand{\one}[0]{\mathds{1}}

\newcommand{\ket}[1]{|#1\rangle}
\newcommand{\braket}[2]{\langle #1|#2\rangle}
\newcommand{\ketbra}[2]{| #1\rangle \langle #2|}
\newcommand{\dyad}[1]{| #1\rangle \langle #1|}
\newcommand{\plus}{\mbox{$+$}}  			   % Nice minus and plus signs, to be used within kets

\renewcommand{\r}{\varrho}

\newcommand{\GHZ}{\text{GHZ}}

\begin{document}

\title		{Some Ulam's reconstruction problems for quantum states}
\date	{\today}

\author{Felix Huber\thanks{ICFO - The Institute of Photonic Sciences, Barcelona, Spain; Theoretical Quantum Optics, University of Siegen, Germany. Email: felix.huber@icfo.eu} \and Simone Severini\thanks{Department of Computer Science, University College London, UK; Institute of Natural Sciences, Shanghai Jiao Tong University, China}}
\maketitle

\begin{abstract}
  Provided a complete set of putative $k$-body reductions of a multipartite quantum state, can one determine if a joint state exists? 
We derive necessary conditions for this to be true. In contrast to what is known as the quantum marginal problem, we consider a setting where the labeling of the subsystems is unknown. The problem can be seen in analogy to Ulam's reconstruction conjecture in graph theory. The conjecture -- still unsolved -- claims that every graph on at least three vertices can uniquely be reconstructed from the set of its vertex-deleted subgraphs. When considering quantum states, we demonstrate that the non-existence of joint states can, in some cases, already be inferred from a set of marginals having the size of just more than half of the parties. We apply these methods to graph states, where many constraints can be evaluated by knowing the number of stabilizer elements of certain weights that appear in the reductions. This perspective links with constraints that were derived in the context of quantum error-correcting codes and polynomial invariants. Some of these constraints can be interpreted as monogamy-like relations that limit the correlations arising from quantum states. Lastly, we provide an answer to Ulam's reconstruction problem for generic quantum states.
\end{abstract}

%%%% Motivation
%%%%%%%%%%%%%%%%%%%%%%%%%%%%%%%%%%%

\section{Introduction}
The relationship of the whole to its parts lies at the heart of the theory of quantum entanglement. 
A pure quantum state is said to be \emph{entangled} if it can not be written as the tensor product of its reductions. A particularly intriguing and important consequence of this mathematical definition is, that given a set of quantum marginals, it is not clear from the outset if and how they can be assembled into a pure joint state. Understanding this problem is not only important in the theory of entanglement, but also for applications in solid-state physics and quantum chemistry, such as calculating the energies of ground states~\cite{RevModPhys.35.668, NatResCounc1995}.

These kind of reconstruction questions have a long tradition in the mathematics literature. From our side, a particularly interesting context is the 1960's {\em Ulam graph reconstruction conjecture}~\cite{AIGNER19943, 10.2307/2316851, JGT:JGT3190010306, Bondy1978}. Indeed, since graphs as well as quantum states can both be represented by positive semi-definite matrices -- the Laplacian and the density matrix respectively -- this highlights a common theme. Moreover, it may be valuable to notice that in the quantum mechanical setting, relational information is associated with correlations between subsystems. 
The family of graphs states allows to directly encode such relational information in pure quantum states. Using approaches from quantum mechanics, this may provide new insights into certain aspects of graph theory.

Informally, the Ulam reconstruction conjecture is as follows: given a complete set of vertex-deleted subgraphs, is the original graph (without vertex labels) the only possible joint graph? 
Over the last decade a substantial amount of research focused on this problem. Widely believed to be true, this remains one of the outstanding unsolved questions in graph theory. 

In this work, we start by providing background to the so-called {\em quantum marginal problem} (QMP), which asks analogous questions for the case that the labels of the individual subsystems are known~\cite{RevModPhys.35.668, quant-ph/0409113v1, Schilling2014_Thesis}. This is in contrast to the question which we will investigate later, namely a setting in which the labels are unknown to us.
Originally coined the {\em $N$-representability problem} by Coleman, its first formulation
asks how to recognize when a putative two-party reduced density matrix is in fact 
the reduction of an $N$-particle system of indistinguishable Fermions~\cite{RevModPhys.35.668}. 
In fact, the $N$-representability problem has been highlighted as one of the most
prominent research challenges in theoretical and computational chemistry~\cite{NatResCounc1995}.
This question was subsequently expanded to the case of distinguishable particles, 
and in particular, to qubits.
In the case of the marginals being disjoint, the conditions for the existence of a 
joint state have been completely characterized:
considering the existence of a pure joint qubit state, the characterization is given by the 
so-called Polygon inequalities, which constrain the spectra of pure state reductions~\cite{PhysRevLett.90.107902}. 
Constraints for the existence of a mixed joint state on two qubits 
have been subsequently been obtained by Bravyi~\cite{journals/qic/Bravyi04}. Solving the QMP in case of 
disjoint marginals completely, Klyachko extended the spectral conditions to the existence 
to a mixed joint state on \(n\) parties of arbitrary local dimensions~\cite{quant-ph/0409113v1}.

The QMP problem in the case of overlapping marginals has turned out to be an even harder problem.
Only few necessary conditions for the general case are known~\cite{Butterley2006, 0806.2962v1, doi:10.1063/1.4808218, PhysRevA.93.032105}, 
of which many are based on entropic inequalities such as the strong subadditivity.
Other constraints are posed by monogamy (in)equalities~\cite{PhysRevA.61.052306, PhysRevLett.96.220503, PhysRevLett.114.140402}, 
some of which will also be used in our work.
Interestingly, the special case of the symmetric extension of two qubits, 
where a two-party density matrix \(\r_{AB}\) is extended to a tripartite state \(\r_{ABB'}\) 
with \(\r_{AB} = \r_{AB'}\), has completely been characterized~\cite{PhysRevA.90.032318}. 
Despite many efforts, a general necessary and sufficient condition for the QMP with
overlapping marginals is still lacking.

A question related to the QMP are the conditions for the uniqueness of the joint state given its marginals. 
This is motivated by a naturally arising physical question:
considering a Hamiltonian with local interactions only, its groundstate is non-degenerate only if no other states with the same local reductions exist.
In this context, Linden et al. showed that almost every pure state of three qubits is completely determined by its two-particle 
reduced density matrices~\cite{PhysRevLett.89.207901, PhysRevA.70.010302}. This result has been subsequently been expanded to 
systems of \(n\) qudits, where having access to a certain subset of all marginals of size \(\lceil n/2 \rceil +1\) 
is almost always sufficient to uniquely specify a joint pure state~\cite{PhysRevA.71.012324}. 
Furthermore, almost all pure states are uniquely determined (amongst pure states) by a set
of three carefully chosen marginals of size \((n-2)\)~\cite{PhysRevA.96.010102}. 
Finally, it is useful to remark that, while the QMP can in principle be stated as a semidefinite program~\cite{Fukuda2007},
its formulation scales exponentially in system size. 
In fact, the QMP has been shown to be QMA-complete~\cite{Liu2006}.

In contrast to previous work on the QMP, we consider in this work only {\em unlabeled} marginals, 
that is, marginals whose corresponding subsystems are unknown to us.  
Thus, one is free to arrange them as necessary in order to obtain a joint state.
Should the reductions to one party be all different (e.g. when considering reductions of random states), 
such labels can naturally be restored by comparing the one-body reductions. 
However, we are here mainly considering a special type of quantum states called graph states. 
These have proven to be useful for certain tasks in quantum information such as 
quantum error correction~\cite{PhysRevA.65.012308, 1023317} and 
measurement-based quantum computation~\cite{PhysRevLett.86.5188, PhysRevA.68.022312,PhysRevA.69.062311}.
The few-body reductions of this type of states are typically maximally mixed, so the strategy of comparing one-body reductions 
does not find an immediate application.
Thus the quantum marginal problem amounts to a kind of jigsaw puzzle: 
we are given overlapping parts, the task being to determine whether or not they indeed can be assembled to one or many different puzzles.

Here, we address similar questions in the case of unlabeled marginals 
and derive necessary constraints for the Ulam reconstruction problem for quantum graph states. 
These are based solely on the number of so-called stabilizer elements present in the complete set of 
reductions having a given size. Our results connect with constraints that were derived in the context 
of quantum error-correcting codes that involve polynomial invariants. These can be interpreted as 
monogamy-like relations that limit the correlations that can arise from quantum states.

The remainder of this paper is organized as follows. In Sec.~\ref{eq:class_Ulam}, we introduce the classical Ulam conjecture.
In Secs.~\ref{sect:setup} and~\ref{sect:graph_states}, we introduce the basic notions of many-qubit systems and graph states that will be useful in our context 
and formulate quantum analogues of Ulam type reconstruction problems.
The main tool of this paper, the so-called weight distribution, is introduced in Sec.~\ref{sect:weight_distribution}.
We derive constraints on the weight distribution in Sec.~\ref{sect:weight_distr_constraints}. These are then applied in Sec.~\ref{sec:detect_illegitimate_decks} to state legitimacy conditions on marginals to originate from a putative joint state. We conclude and provide an outlook in Sec.~\ref{sect:conclusion}.

%%%% REALIZABILITY AND UNIQUENESS
%%%%%%%%%%%%%%%%%%%%%%%%%%%%%%%%%%%

\section{Deck legitimacy and graph reconstruction}\label{eq:class_Ulam}
Consider a simple graph \(G = (V,E)\) on~\(n\) vertices. 
Denote by~\(N(i)\) the neighborhood of vertex~\(i\), 
that is, the vertices adjacent to~\(i\). 
By deleting a single vertex \(j \in V\) and deleting each edge~\(e\) incident with~\(j\), 
one obtains the vertex-deleted subgraph \(G_j = (V\backslash \{j\}, E \backslash \{e\} \,|\, j \in e)\) 
on \((n-1)\) vertices. 
By forming all vertex-deleted subgraphs \(G_j\) induced by~\(G\), termed \emph{cards},
we obtain its unordered \emph{deck}, the multi-set \(\DD_G = \{G_1, \dots, G_n\}\).
We denote by \(\binom{G}{F}\) the number of times a copy of graph \(F\) appears as a subgraph of \(G\).

Concerning graphs and their decks, we present three closely related problems that originate in a conjecture by Ulam, 
the so-called {\em reconstruction conjecture}. 
A wealth of results concern themselves with the reconstruction conjecture and related questions; we refer to Ref.~\cite{lauri_scapellato_2016}, to Sec.~$2.3$ in Ref.~\cite{Gross:2013:HGT:2613412}, and to Sec.~$2.7$ in Ref.~\cite{BondyMurty2008} for a full account.

\vspace{0.8em}
\noindent{\bf A) Legitimate deck problem.} 
Let a putative deck be given that contains~\(n\) cards of size \((n-1)\) each. 
Is it possible to obtain it from a graph of~\(n\) vertices? 
If this is indeed the case, the deck is called {\em legitimate}. 
The legitimate deck problem is a type of realizability problem~\cite{AIGNER19943}, 
where it has to be determined if a graph realizing a certain property exists, or in other words, if that property is {\em graphical}.

Bondy provided a necessary condition for a deck to be legitimate, called Kelly's condition. It already seems to detect a majority of illegitimate decks.

%%%% Kelly Conditions
\begin{proposition}[Kelly's condition~\cite{Bondy1978, BondyMurty2008}]
\label{Kelly}
  Let \(\DD\) be a complete deck of a putative graph of \(n\) vertices. 
  Given a graph~\(F\) with a number of vertices \(|V_F|< n\), the following expression must be an integer 
  \begin{equation}\label{eq:Kelly_cond}
  	\frac{\sum_{i=1}^n \binom{G_i}{F}}{n-|V_F|}\,.
  \end{equation}
\end{proposition}
\begin{proof}
 Each instance of a copy of \(F\) in \(G\) occurs exactly \(|n| - |V_F|\) times in the deck \(\DD_G\). Thus \(\sum_{i=1}^n \binom{G_i}{F} / (n-|V_F|)\) must be integer.
\end{proof}

Thus whenever Eq.~\eqref{eq:Kelly_cond} is not an integer, the deck is illegitimate and can not be realized by a graph.

Despite the innocuous looking formulation, 
the legitimate deck problem is at least as hard as that of graph isomorphism~\cite{mansfield82, Harary82}. 
Not much more seems to be known about the legitimate deck problem, apart from the fact  
that many counting arguments used to approach the reconstruction conjecture (as the one above) 
can be used to derive natural legitimacy conditions~\cite{BondyMurty2008}.

\vspace{0.8em}
\noindent{\bf B) Reconstruction problem (uniqueness).}
Given a legitimate deck \(\DD_G\), is there, up to graph isomorphism, a unique graph corresponding to it? In other words, is it {\em reconstructible}?
The Ulam graph reconstruction conjecture states that this must indeed be the case for all graphs on at least three vertices.
\begin{conjecture}[Ulam~\cite{Ulam1960,kelly1957,Harary1964}]
    Considering graphs with at least three vertices, \(\DD_G = \DD_H\) if and only if~\(G\) is isomorphic to~\(H\).
\end{conjecture}

The conjecture has been verified for a number of special cases: amongst other results, regular graphs, disconnected graphs, trees, separable graphs, and all graphs on nine or fewer vertices are reconstructible~\cite{Gross:2013:HGT:2613412}.
Furthermore, Bollob\'as has shown that three carefully chosen cards are sufficient to uniquely reconstruct the original graph for most decks~\cite{JGT:JGT3190140102}, that is, most graphs have a {\em reconstruction number} \(\operatorname{rn}(G) = 3\). This type of asymptotic result indicates that most graphs are easy to specify, in analogy with the Weisfeiler-Lehman method for graph isomorphism~\cite{Shervashidze2011}. 

Let us stress that the reconstruction conjecture is not about an finding an efficient algorithm to obtain an explicit reconstruction of a graph, but rather about the uniqueness of the reconstruction. 

\vspace{0.8em}
\noindent{\bf C) Reconstruction problem (constructive).}
Accordingly, we have the following category concerned with finding an explicit reconstruction:
given a legitimate deck, provide an efficient procedure to reconstruct a compatible graph.
\vspace{0.8em}

In what follows, we aim to treat the above Ulam's reconstruction problems 
for a special type of quantum states called {\em graph states}: 
given a collection of graph state marginals, we ask for the reconstruction of a corresponding joint state,
inquire about its uniqueness, and aim to develop methods to decide the legitimacy of the marginal deck.

%%%% Set up
%%%%%%%%%%%%%%%%%%%%%%%%%%%%%%%%%%%

\section{Set up}\label{sect:setup}
A few definitions concerning multipartite quantum states are in order. Denote by \(I, X, Y\), and \(Z\) the identity and the three Pauli matrices
\begin{align}
I &= \begin{pmatrix}
      1 & 0 \\
      0 & 1 
	\end{pmatrix} \,, 
&X &= \begin{pmatrix}
      0 & 1 \\
      1 & 0 
	\end{pmatrix} \,, 
&Y &= \begin{pmatrix}
      0 & -i \\
      i & 0 
	\end{pmatrix} \,,
&Z &= \begin{pmatrix}
      1 & 0 \\
      0 & -1 
	\end{pmatrix} \,.
\end{align}
The single-qubit Pauli group is defined as
\( \GG_1 = \langle i, X,Y,Z \rangle\), 
from which the \(n\)-qubit Pauli group is constructed by its $n$-fold tensor product
\( \GG_n = \GG_1 \ot \cdots \ot \GG_1\) 
\,\,($n$ times).
By forming tensor products of Pauli matrices, we obtain an orthonormal basis \(\P = \{P\}\)
of Hermitian operators acting on \((\mathbb{C}^2)^{\ot n}\), with \(\tr(P_\alpha P_\beta) = \delta_{\alpha\beta}2^n\).
We will write \(X_j, Y_j, Z_j\) for Pauli matrices acting on particle \(j\) alone.
Denote by \(\supp(P)\) the support of an operator \(P\in \P\), that is, the parties on which~\(P\) acts non-trivially with \(X\), \(Y\), or \(Z\). 
The weight of an operator is then the size of its support, \(\wt(E) = |\supp(E)|\).

Pure quantum states \(\ket{\psi}\) of \(n\) qubits are represented by unit 
vectors in \((\C^2)^{\ot n}\). Their corresponding 
density matrices~\footnote{{\em Mixed} states~\(\r\) are represented by positive-semidefinite Hermitean operators that are of unit trace but not of rank one, 
i.e. they are statistical mixtures of pure states: \(\r = \sum_i p_i \dyad{\psi_i}\) with \(p_i\geq 0\) and \(\sum_i p_i = 1\).}
\(\r = \dyad{\psi}\) can in terms of Pauli matrices be expanded as
\begin{equation}\label{eq:bloch_represenation}
	\r = 2^{-n} \sum_{P \in \P} \tr[P^\dag \r] P \,.
\end{equation}
Given a quantum state \(\r\), we obtain its marginal (also called reduction) on a subset \(A\) by acting with the partial trace on its complement \(A^c\),
\(\r_A = \tr_{A^c}(\r) \).
In the Bloch decomposition [Eq.~\eqref{eq:bloch_represenation}], 
the reduction onto subsystem \(A\) tensored by the identity on \(A^c\) can also be written as
\begin{equation}
	\r_A \ot \one_{A^c} \,\,= \!\!\sum_{\supp(P) \subseteq A} \tr[P^\dag \r] P \,.
\end{equation}
This follows from \(\tr_{A^c}(E) = 0\) if \(\supp(E) \not\subseteq A\).

A pure multipartite state is called {\em entangled}, if it cannot be written as the tensor product of single-party states,
\begin{equation}
  \ket{\psi}^{\text{ent}} \neq \ket{\psi_1} \ot \ket{\psi_2} \ot \cdots \ot \ket{\psi_n} \,.
\end{equation}
For a mixed state, entanglement is present if it cannot be written as a convex combination of product states,
\begin{equation}
 \varrho^{\text{ent}} \neq \sum_i p_i \varrho_1 \ot \varrho_2 \ot \cdots \ot \varrho_n \,,
\end{equation} 
for all single-party states \(\varrho_1 \dots \varrho_n \) and probabilities \(p_i\) with \(\sum_i p_i = 1\), \(p_i \geq 0\).
Note that for a pure product state, all reductions are pure too, having the purity \(\tr(\r_A^2) = 1\).

Lastly, we note that any pure state can, for every bipartition \(A|B\), be written as
\begin{equation}
 \ket{\psi_{AB}} = \sum_i \sqrt{\lambda_i} \ket{i_A} \ot \ket{i_B}\,,
\end{equation} 
where \(\{\ket{i_A}\}\) and \(\{\ket{i_B}\}\) are orthonormal bases for subsystems \(A\) and \(B\), and \(\sqrt{\lambda_i}\geq 0\) with \(\sum_i \lambda_i = 1\).
This is called the {\em Schmidt decomposition}.

We are now able to introduce the analogue of a graph deck for quantum states.
 \begin{definition}
     A \emph{quantum \(k\)-deck} is a collection of quantum marginals, 
     termed (quantum) cards, of size \(k\) each. The marginals are unlabeled and thus not associated to any specific subsystems.
     A~deck is called \emph{complete} if it contains \(\binom{n}{k}\) cards
     and \emph{legitimate} if it originates from a joint state.
 Given a quantum state \(\r\), its corresponding \(k\)-deck is given 
 by the collection of all its marginals of size \(k\),
 \begin{equation}
 	\DD_{\r} = \{ \r_A \,|\,\, |A| = k , A \subseteq \{1\dots n\}\}\,.
 \end{equation}
\end{definition}

With these definitions in place we can ask Ulam's reconstruction type questions about decks formed from quantum marginals:

\vspace{0.8em}
\noindent{\bf A) Legitimate deck problem.} 
A putative quantum deck be given, is it possible to obtain it from a (pure or mixed) quantum state of \(n\)~parties? 

\vspace{0.8em}
\noindent{\bf B) Reconstruction problem (uniqueness).}
Given a legitimate quantum \(k\)-deck~\(\DD_\r\), is there, up to particle relabeling, a unique quantum state corresponding to it? For a given state~\(\r\), what is the size of the smallest subdeck such that a unique reconstruction amongst pure (mixed) states is feasible, i.e., what is its {\em reconstruction number}~\(\operatorname{rn}_{\text{pure (mix)}}(\r)\)?

\vspace{0.8em}
\noindent{\bf C) Reconstruction problem (constructive).}
Given a legitimate quantum deck, provide an efficient procedure to reconstruct a compatible (pure or mixed) quantum state. 
Such type of problem is important in e.g. quantum state tomography, 
where often only few-body correlations together with their particle labels are known.
\vspace{0.8em}

%%%% Graph states
%%%%%%%%%%%%%%%%%%%%%%%%%%%%%%%%%%%

\section{Graph states}\label{sect:graph_states}
To approach Ulam type problems in the quantum setting, let us introduce {\em graph states}. 
These are a type of pure quantum states which are completely characterized by corresponding graphs. 
Of course we could also approach a more general setting by considering generic pure states. 
Our choice is suggested by the immediate connection between graph states and graphs. 
Our hope is that the mathematical richness of graph states could highlight some new perspective 
on Ulam's problem, even when we restrict our attention to graph states only. 

%%%% DEFINITION
\begin{definition}[\cite{PhysRevA.69.062311}]
  Given a simple graph \(G = (V,E)\) having \(n\) vertices, its corresponding \emph{graph state} \(\ket{G}\) is defined as
  the common $(+1)$-eigenstate of the $n$ commuting operators \(\{g_i\}\),
    \begin{equation}
      g_i = X_i \bigotimes_{j \in N(i)} Z_j \,.
    \end{equation}
  Thus \(g_i \ket{G} = \ket{G}\) for all \(g_i \). The set \(\{g_i\}\) is called the \emph{generator} of the graph state.
\end{definition}
%%%% END DEFINITION

To obtain \(\dyad{G}\) explicitly, the notion of its {\em stabilizer} is helpful. The stabilizer \(S\) is the Abelian group obtained by
the multiplication of generator elements,
\begin{equation}
	S = \Big\{ s = \prod_{i \in I}  g_i \,\, | \,\, I \subseteq \{1,\dots, n\}\Big\} \,.
\end{equation}
Each of its \(2^n\) elements stabilize the state, \(g_i \ket{G} = \ket{G}\) for all \(g_i \).
Naturally, the stabilizer forms a subgroup of the 
\(n\)-party Pauli-group which consists of all elements in \(\P\) in addition
to a complex phase \(\{\pm 1,\pm i\}\).
With this the graph state reads~\cite{PhysRevA.69.062311}
\begin{equation}
	\dyad{G} = \frac{1}{2^n} \sum_{s \in S} s \,.
\end{equation}
On the other hand, it can be shown that the graph state can also be written as
\begin{equation}\label{eq:graph_state_C}
  \ket{G} = \prod_{e \in E} C_{e} \ket{\plus}_{V} \,,
\end{equation}
where \(\ket{\plus}_V = \bigotimes_{j\in V} (\ket{0}_j + \ket{1}_j)/\sqrt{2}\), 
and the controlled-\(Z\) gate acting on the parties in edge \(e=(i,j)\) reads \(C_{e} = \operatorname{diag}(1,1,1-1)\).

In this picture, it is evident that graph states can be described as real equally weighted states:
by initializing in \(\ket{+}_V\), 
an equal superposition of all computational basis states is created. 
The subsequent application of \(C_e\) gates then only changes certain signs in this superposition.

To understand when two non-isomorphic graphs give different but 
comparable quantum states, let us take a small detour. The first step is to clarify the meaning of {\em comparable}. 
When looking at similarities between quantum states, the equivalence up to local unitaries (LU) 
is often considered. Two \(n\)-qubit states \(\sigma\) and \(\r\) are said to be \emph{LU-equivalent}, if there exist unitaries \(U_1, \dots, U_n \in SU(2)\), such that
\( \sigma = U_1 \ot \cdots \ot U_n \,\, \r \,\, U_1^\dag \ot \cdots \ot U_n^\dag \).
If no such matrices exist the states are said to be {\em LU-inequivalent}. An interesting subset of unitaries to consider is the so-called local Clifford group~\(\CC_n\).
It is obtained by the n-fold tensor product of the one-qubit Clifford group \(\CC_1\),
\begin{equation}
 \CC_1 = \Big<  \frac{1}{\sqrt{2}} 
    \begin{pmatrix}
      1 & 1 \\
      1 &-1
    \end{pmatrix},
    \begin{pmatrix}
      1 & 0 \\
      0 & i
     \end{pmatrix} \Big> \,.
\end{equation}
The group \(\CC_1\) maps the one-qubit Pauli group \(\GG_1 = \langle i, X,Y,Z\rangle\) to itself under conjugation.
The \(n\)-qubit local Clifford group \(\CC_n = \CC_1 \ot \dots \ot \CC_1\) (\(n\) times)
then similarly maps the \(n\)-qubit Pauli group \(\GG_n\) to itself under conjugation.
Interestingly, it was shown that the action of local Clifford operations on a graph state
can be understood as a sequence of local complementations on the corresponding graph~\cite{PhysRevA.69.022316}.
This works in the following way: given a graph~\(G\), its local complementation with respect 
to vertex~\(j\) is defined as the complementation of the subgraph in~\(G\) consisting 
of all vertices in its neighborhood~\(N(j)\) and their common edges.
We conclude that if two graphs are in the same local complementation orbit, 
then their corresponding graph states must be equivalent under the action of 
local Clifford operations, and vice versa. Thus they must also be LU-equivalent.
The contrary however is not necessarily true. Indeed,
it has been shown that there exist LU-equivalent graph states which are not
local Clifford equivalent~\cite{journals/qic/JiCWY10, 1751-8121-50-19-195302}.

We begin our analysis with an observation concerning the reductions 
of graph states onto \(n-1\) parties~\cite{1751-8121-48-9-095301}.
For this, let us define a {\em vertex-shrunken} graph:
the vertex-shrunken graph~\(S_i\) is obtained by deleting vertex~\(i\) and by shrinking all of its incident edges~\((i,j)\) to so-called one-edges~\((j)\). For a simple graph, the operation simply marks all vertices adjacent to~\(i\). To clarify the meaning of this notion, it is useful to look at a generalization to hypergraphs. Hypergraphs can have edges containing more than two vertices. Consider now deleting a single vertex~\(i\) from a hypergraph: an incident hyperedge \(e\) can, instead of simply being discarded, be {\em shrunken} such as to still contain all remaining vertices. The shrunken edge then reads \(e \backslash \{i\}\). In that way, a \(k\)-edge, initially connected \(k\)~vertices, becomes a \((k-1)\)-edge. Consequently, shrinking a \(2\)-edge yields a one-edge. One has
\begin{align}
S_i = (V \backslash {i}, 
\{ e\backslash {i} \, |\, i \in e\} \cup
\{e \,|\, i \notin e\})\,,
\end{align}
and the notion of a one-edge is well-motivated.
The following Proposition, originally obtained by Ref.~\cite{PhysRevA.69.062311}, is concerned with expressing reductions of graph states in terms of vertex-deleted and -shrunken graphs.

%%%% PROPOSITION
\begin{proposition}[Lyons et al.~\cite{PhysRevA.69.062311}]
Consider the quantum \((n-1)\)-deck of a graph state~\(\ket{G}\). Then, each of its cards can be represented by two graphs:
a vertex-deleted graph \(G_j\) and a vertex-shrunken graph \(S_j\), each having~\((n-1)\) vertices.
\end{proposition}
%%%% END PROPOSITION

%%%% BEGIN PROOF
\begin{proof}
In Eq.~\eqref{eq:graph_state_C}, let us single out vertex~\(j\) to be traced over.
\begin{align}
  \ket{G} &= \prod_{e\in E} C_e \ket{+}_V
	  = \left( \ket{0}_j \,+\! \prod_{e \in E \,|\, j \in e}\!\!\! C_{e \backslash \{j\} } \ket{1}_j  \right)  
	  \,\,\ot \!\prod_{e' \in E \,|\, j \notin e'} \!\!\!C_{e'} 
 \ket{+}_{V \backslash \{j\}}\,.
\end{align}
Note that if \(C_e\) is a controlled \(Z\)-gate acting on parties~\(i\) and~\(j\),
then \(C_{e\backslash\{j\}}\) is the local \(Z_i\) gate acting on party~\(i\) alone.
Thus performing a partial trace over subsystem \(j\) yields
\begin{align}
  \tr_j [\dyad{G}]  &= \langle 0_j \dyad{G} 0_j\rangle + 
		       \langle 1_j \dyad{G} 1_j \rangle \nn\\
  &= \frac{1}{2} \Big(\underbrace{
	\prod_{e \in E \,|\, j \notin e} \!\!\!C_{e}\,
	(\dyad{+})_{V \backslash \{j\}} 
	\prod_{e' \in E \,|\, j \notin e'} \!\!\!C_{e'}
	}_{\rm{delete}}   \nn\\
   &+ \underbrace{ 
	\prod_{i\in N(j)} \!\!Z_i
	\prod_{e \in E \,|\, j \notin e} \!\!\!C_{e} \,
	(\dyad{+})_{V \backslash \{j\}}  
	\prod_{e'\in e \,|\, j \notin e'} \!\!\!C_{e'}
	\prod_{i\in N(j)} \!\!Z_i 
	}_{\rm{shrink}} \Big) \,.
\end{align}
The reduction of a graph state onto \((n-1)\) parties is thus given by the equal mixture of two graph states: 
a vertex-deleted graph state \(\ket{G_i}\), whose graph is the vertex-deleted subgraph of~\(G\), 
and a vertex-shrunken graph state \(\ket{S_j}\), whose graph is a vertex-deleted subgraph 
with additional one-edges on \(N(j)\) caused by shrinking all edges adjacent to~\(j\). These one-edges correspond to 
local \(Z_j\)-gates. One obtains
\begin{align}
	  \ket{G_j} &= \prod_{e\in E \,|\, j \notin e} C_{e} \ket{+}_{V \backslash \{j\}}\,, &
	  \ket{S_j} &= \prod_{i\in N(j)} Z_i  \prod_{e\in E \,|\, j \notin e}\!\!\! C_{e} \ket{+}_{V \backslash \{j\}} \,,
\end{align}
and we can write
  \(
      \tr_j(\dyad{G}) = \frac{1}{2}(\dyad{G_j} + \dyad{S_j}) \,. \label{eq:red_graph_state}
  \)
This ends the proof.
\end{proof}
%%%% END PROOF
%
If the graph \(G\) is fully connected, then \( \braket{G_j}{S_j} = 0\) for all~\(j\). This follows from the fact
that all stabilizer elements corresponding to a fully connected graph must have weights larger or equal than two.
Thus the one-body reductions are maximally mixed, and the complementary \((n-1)\)-body reductions
must be proportional to projectors of rank two. When tracing out more than one party, this procedure of substituting each graph by the equal mixture of its vertex-deleted and vertex-shrunken subgraphs is iteratively repeated. Thus the reduction of a graph state of size \(n-k\) is represented by
a collection of \(2^k\) graphs.

%%%% full state reconstructible from single card in computational basis, up to local Z.

Let us now consider a specific formulation of the Ulam graph problem in the quantum setting
where all \((n-1)\)-body reductions of a graph state are given in the computational basis.
What can one say about the joint state? As shown by Lyons et al.~\cite{1751-8121-48-9-095301}, 
it turns out that the joint state can - up to a local \(Z\) gate - be reconstructed from a single card.

%%%% PROPOSITION
\begin{proposition}[Lyons et al.~\cite{1751-8121-48-9-095301}]
\label{prop:rec_from_one_card}
    Given a legitimate \((n-1)\)-deck of a graph state \(\ket{G}\) in the computational basis, 
    the joint state \(\ket{G}\) can be reconstructed up to local \(Z_j\) gates from any single card.
\end{proposition}
%%%% END PROPOSITION

%%%% BEGIN PROOF
\begin{proof}
Let us expand the graph states \(\ket{G_j}\) and \(\ket{S_j}\) as
appearing in~\eqref{eq:red_graph_state} in the computational basis.
Due to our ignorance about the joint state, denote them by \(\ket{\alpha}\) and
\(\ket{\beta}\), where either one could be the vertex-deleted graph state,
with the other one being the vertex-shrunken graph state. 
From Eq.~\eqref{eq:graph_state_C} follows that all graph states are real equally weighted states.
Thus it is possible to expand 
\begin{align}
	&\ket{\alpha} = \frac{1}{\sqrt{N}} \sum_{i=0}^{N-1} \alpha_i \ket{i}\,, &
	&\ket{\beta} =  \frac{1}{\sqrt{N}} \sum_{i=0}^{N-1} \beta_i  \ket{i} \,,
\end{align}
where \(N = 2^n\), \(\alpha_i, \beta_i \in \{ -1,1\} \), and \(\{\ket{i}\}\) being the computational basis for \(V\backslash \{j\}\).
We can therefore write the card \(C^k = \tr_k(\dyad{G})\) as
\begin{equation}
	C^k = \frac{1}{2N}\sum_{i,j=0}^{N-1}(\alpha_i\alpha_j + \beta_i\beta_j) \ketbra{i}{j} \,.
\end{equation}
Because of \(\alpha_i, \beta_i \in \{-1,1\}\), 
\(2N C^k_{ij}\) can only be \(0\) or \(\pm 1\).
Because \(\ket{0\dots 0}\) remains unaffected by conditional phase gates, \(\alpha_1 = \beta_1 = 1\).
Furthermore, \(\alpha_j = \beta_j = \on{sign}(C^k_{1j})\) for all \(j\) where \(C^k_{1j} \neq 0\). 
On the other hand, when \(C^k_{1l} = 0\), then \(\alpha_l = -\beta_l\). 
Without loss of generality, set \(\alpha_m = -\beta_m = 1\) for the first instance of \(m\) where this happens. 
The remaining but yet undetermined coefficients \(\alpha_l = -\beta_l \in \{\pm 1\}\) are given from
the entries \(C^k_{ml}\),
\begin{align}
  \alpha_m \alpha_l + \beta_m \beta_l &= \alpha_l - \beta_l = 2 \alpha_l = 2N C^k_{ml}\,.
\end{align}
This completely determines the remaining coefficients of \(\ket{\alpha}\) and \(\ket{\beta}\). 
Now the task is to reconstruct the graphs corresponding to \(\ket{\alpha}\) and \(\ket{\beta}\).
This can be done by iteratively erasing all minus signs in the computational basis expansion~\cite{1367-2630-15-11-113022}:
first, minus signs in front of terms having a single excitation only, 
e.g. \(\ket{0 \dots 010 \dots \dots 0}\), are removed by local \(Z_j\) gates.
Then, conditional phase gates are applied to erase minus signs in front of 
terms having two excitations, and so on. 
By this procedure, one obtains the state \(\ket{+}^{\ot n}\) and all the gates necessary to
obtain the original graph state, thus determining the graph.

Note that the symmetric difference of the two graphs corresponding to \(\ket{\alpha}\) and \(\ket{\beta}\) yields 
all edges that were severed under the partial trace operation,
\begin{equation}
	\underbrace{  \prod_{e\in E \,|\,  j \in e} C_{e \backslash \{j\} }  
		      \prod_{e'\in E \,|\, j \notin e'} \!\!\!C_{e'}  }_{\text{shrink}} 
	\,\,\underbrace{  \prod_{e''\in E \,|\, j \notin e''} \!\!\!C_{e''}  }_{\text{delete}}
	\,\,= \!\!\underbrace{\prod_{e\in E \,|\, j \in e}}_{\text{edges connected to $j$}}  \!\!\!\!\!\!\!\!\!C_{e \backslash \{j\}}
	\,\,= \prod_{i \in N(j)} \!Z_j \,.
\end{equation}
The original graph state can then only be one of the following
\begin{align}
	\dyad{G} &= \prod_{e\in E \,|\, j \in e} \!\!\!C_{e} \,\ket{\alpha} \otimes \ket{+}_j  \,,\quad \text{ or} \nonumber\\
	\dyad{G} &= \prod_{e\in E \,|\, j \in e} \!\!\!C_{e} \,\ket{\beta}  \otimes \ket{+}_j  \,.
\end{align}
This proves the claim.
\end{proof}
%%%% END PROOF

Let us compare this to the classical case: there three cards suffice to reconstruct most graphs~\cite{JGT:JGT3190140102}.
The fact that graph states can - up to local $Z$-gates - be reconstructed from a single marginal of size \((n-1)\) illustrates how the state's density matrix inherently contains more information about the original graph than its classical counterpart the adjacency matrix. This is to be expected: the former is of size \(2^n \times 2^n\), while the latter has dimensions \(n\times n\) only.

%%%% Weight distribution
%%%%%%%%%%%%%%%%%%%%%%%%%%%%%%%%%%%

\section{The weight distribution}\label{sect:weight_distribution}
In order to determine whether or not, given a quantum \(k\)-deck, 
a joint graph state could possibly exist, we introduce 
the {\em weight distribution} of quantum states. 
This is a tool from the theory of quantum error correction 
that can be used to characterize the number of errors a code can correct.
Parts of the weight distribution can be obtained from a complete quantum \(k\)-deck already,
and no knowledge of the labeling of the individual parties is needed, making this tool useful for legitimate deck type problems.
It turns out that the weight distribution can indeed signal the illegitimacy of certain decks, 
that is, collections of marginals that are incompatible with any joint state.

%%%% DEFINITION
\begin{definition}[\cite{681316, PhysRevLett.78.1600, Aschauer2004}]
The weight distribution of a multipartite qubit quantum state \(\r\) is given by
\begin{equation}
	A_j(\r) = \sum_{\substack{P \in \P \\ \wt(P)=j}} \!\! \tr (P \r) \tr (P^\dag \r) \,,
\end{equation}
where the sum is over all elements \(P\) of weight \(j\) in the \(n\)-qubit Pauli basis \(\P\).
\end{definition}
%%%% END DEFINITION
Note that for higher dimensional quantum systems any appropriate orthonormal 
tensor-product basis can be chosen instead of the Pauli basis, e.g. the Heisenberg-Weyl or Gell-Mann basis.
The weights \(A_j\), being quadratic in the coefficients of the density matrix and invariant under local unitaries, are so-called polynomial invariants of degree two. They characterize the distance of quantum error-correcting codes~\cite{681316, 1751-8121-51-17-175301} and can be used to detect entanglement~\cite{Aschauer2004, PhysRevA.91.042339, PhysRevA.94.042302, PhysRevA.87.034301, 1710.02473v2}.

For graph states, the weight distribution is particularly simple: 
because \(\tr[P\dyad{G}]\) can only be either \(0\) or \(\pm 1\), 
the weight distribution of \(\ket{G}\) is simply given by
the number of its stabilizer elements having weight \(j\),
\begin{equation}
	A_j(\ket{G}) = | \{s \in S | \wt(s)=j \}| \,.
\end{equation}

%%%% BEGIN EXAMPLE
\begin{example}
  The three-qubit graph state corresponding to the 
  fully connected graph of three vertices has the generator 
  \(	G = \{XZZ, ZXZ, ZZX\}\)~\footnote{This is a state that is LU-equivalent to the Greenberger-Horne-Zeilinger state $(\ket{000} + \ket{111})/\sqrt{2}$.}. 
  Its stabilizer reads
  \begin{equation}
    S = \{III, IYY, YIY, YYI, 
    XZZ, ZXZ, ZZX, -XXX\} \,.
  \end{equation}
  Accordingly, its weight distribution is \(A = [A_0, A_1, A_2, A_3] = [1, 0, 3, 4]\). 
  By normalization, \(A_0 = \tr(\r) = 1\) must hold for all states.
  Because \(\r\) is pure, \(\tr(\r^2) = 1\), and thus \(\sum_{j=0}^3 A_j(\ket{\psi}) = 2^3\).
\end{example}
%%%% END EXAMPLE

As a warm-up, let us derive a result on the weight distribution known from quantum error correction~\cite{PhysRevA.69.052330} 
using properties of Pauli matrices only.
%
%%%% PROPOSITION
\begin{proposition}\label{prop:type_I_II}
    Given a graph state, the sum \(\Ae = \sum_{j=0}^{\nhf} A_{2j}\) 
    can only take two possible values,
    \begin{align}
    	\Ae =
    	\begin{cases}
	      2^{n-1}  	&\text{(type I)} \,, \\
	      2^n 		&\text{(type II)}\,.
    	\end{cases}
    \end{align}
\end{proposition}
%
%%%% PROOF
\begin{proof}
    Note that a graph state \(\r = \dyad{G}\) can be decomposed into 
    \begin{align}
    	\r &= \frac{1}{2^n}\Big(
			      \sum_{\substack{P \in \P \\ \wt(P) \text{ even}}} \!\! \tr[P^\dag M] P + 								  
			      \sum_{\substack{P \in \P \\ \wt(P) \text{ odd}}}  \!\! \tr[P^\dag M] P
			     \Big) 
        	 = \frac{1}{2^n}(\Pe + \Po)\,,
    \end{align}
    where \(\Pe\) and \(\Po\) are the sums of all stabilizer elements having even and odd weight respectively.
    Because of \(s \r = \r \) for all \(s\in S\), also \(\Pe\) and \(\Po\) have \(\r\) as an eigenvector.
    We apply this decomposition to \(\r = \r^2\), making use of Lemma $1$ from Ref.~\cite{PhysRevLett.118.200502}
    regarding the anti-commutators of elements from~\(\P\):
    the term \(\{\Pe, \Po\}\) appearing in \(\r^2 = \{\r, \r\}/2\) can only contribute to terms of odd weight in \(\r\),
    yielding     \( \{\Pe, \Po\}  = 2^n \Po \).
    Consequently, one obtains
    \begin{equation}
	  \tr ( \{\Pe, \Po\} \r )  = \tr ( 2^n \Po \r )\,.
    \end{equation}
    Accordingly,
    \( 2 \Ae \Ao = 2^{n} \Ao \),
    where \(\Ae\) and \(\Ao\) are the number of terms in the stabilizer that have even and odd weight respectively.
    Consider first \(\Ao \neq 0\). Then \(\Ae = 2^{n-1}\).
    Conversely, if \(\Ao = 0\), then \(\Ae = 2^{n}\),
    because \(\r\) is pure and must thus satisfy \(\sum A_j = \Ae + \Ao = 2^n\).
    This ends the proof.
\end{proof}
The same argument can be done for reductions of graph states that happen to be proportional to projectors of rank \(2^q\).
There, either \(\Ae = 2^{n-q-1}\) or \(\Ae = 2^{n-q}\) holds.

These two cases, that is, graph states of type~\(I\) and type~\(II\), are also known from the theory of 
{\em classical} self-dual additive codes over \(\operatorname{GF}(4)\)~\cite{681315, PhysRevA.69.052330}. 
If only stabilizer elements of even weight are present the code is said to be of type~\(II\), 
while codes having both even and odd correlations in equal amount are of type~\(I\).
It can be shown that all type~\(II\) codes must have even length,
and conversely, self-dual additive codes of odd length~\(n\) are always of type~\(I\).
This is also a direct consequence of the monogamy relation derived in Ref.~\cite{PhysRevLett.114.140402}
which is known to vanish for an odd number of parties, implying \(\Ae = \Ao = 2^{n-1}\).
Let us note that any graph states whose every vertex is connected to an odd number of 
other vertices is of type~\(II\)~\footnote{See Theorem~$15$ in Ref.~\cite{DANIELSEN20061351}.}.
For example, Greenberger-Horne-Zeilinger states of an even number of qubits are of this type, 
being LU-equivalent to fully connected graph states.

\begin{figure}[tbp]\label{fig:cube_hypergraph}
\centering
      \includegraphics[height = 8em]{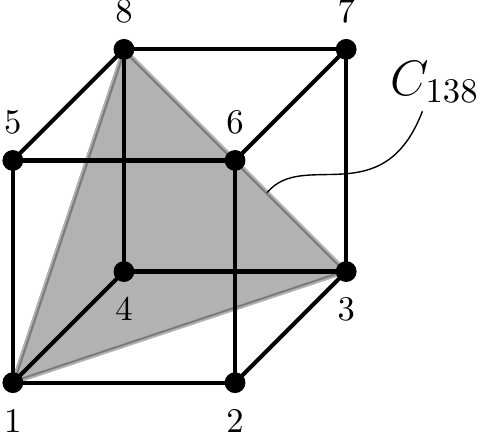}
      \caption[A hypergraph state that is LU-equivalent to graph states.]
      {a hypergraph state that cannot be transformed into a graph state by applying local unitaries. All of its three-body marginals are maximally mixed.}
\end{figure}

This result can be used to show that a particular state cannot be LU-equivalent to any graph state.
Let us consider the state depicted in Fig.~\ref{fig:cube_hypergraph}, which is a so-called {\em hypergraph state}~\cite{1367-2630-15-11-113022}. 
It can be obtained by applying the additional gate \(C_{138} = \operatorname{diag}(1,1,1,1,1,1,1,-1)\) between particles \(1\), \(3\), and \(8\) to the graph state of a cube, \(\ket{H} = C_{138} \ket{G_{\text{cube}}}\).
Its weight distribution reads
\begin{equation}
 A = [1,  0,   0,  0,  30,  48,  96, 48, 33] \,,
\end{equation}
with \(\Ae = \sum_{j \text{ even}} = 160\). This is incompatible with being a graph state of type~\(I\) or type~\(II\),
these having either \(\Ae = 128\) or \(\Ae = 256\) respectively. Because the weight distribution is invariant under LU-operations, 
the state must be LU-inequivalent to graph states.
Let us add that all the three-body marginals of \(\ket{H}\) are maximally mixed, 
and the state can thus be regarded as being highly entangled~\cite{PhysRevA.69.052330}.

%%%% Begin PROPOSITION

Lastly, one could ask whether or not the presence of entanglement can be detected from the weight distribution of a state. 
This is indeed the case.
\begin{proposition}\label{prop:weight_dist_prod}
 Let \(\ket{\psi}\) be a pure product state on $n$ qubits. Then \(A_j(\ket{\psi}) = \binom{n}{j}\).
\end{proposition}
\begin{proof}
Let us assume that we are given a product state on \(m-1\) qubits with weights denoted by \(A_j^{(m-1)}\).
Tensoring it by a pure one-qubit state, the weight \(A^{(m)}_j\) of the resulting state on $m$ qubits reads
  \begin{equation}\label{req:recurrence}
    A_j^{(m)} = A_j^{(m-1)} A_0^{(1)}+ A_{j-1}^{(m-1)} A_1^{(1)}\,.
  \end{equation}
With \(A_0 = A_1 = 1\) for a pure one-qubit state, this matches the recurrence relation that is satisfied by the binomial coefficients, 
 namely
  \begin{equation}
  	\binom{m}{j} = \binom{m-1}{j} + \binom{m-1}{j-1}
  \end{equation}
  together with the initial condition \(A^{(1)}_j = \binom{1}{j} = 1\).
  Thus the weight distribution of a pure product state on \(n\) qubits reads \( A_j(\ket{\psi}) = \binom{n}{j}\) .
  \end{proof}

  Proposition~\ref{prop:weight_dist_prod} is readily seen to generalize to multipartite systems of higher and mixed dimensions: then the initial conditions have to be set accordingly;
  e.g. when considering $D$-level systems, \(A_0^{(1)} = 1\) and \(A_1^{(1)} = (D-1)\) leads to \(A_j(\ket{\psi}) = (D-1)^j \binom{n}{j}\).
 In similar spirit, by considering the purities of the marginals, one sees that if 
  \begin{equation}
\sum_{|S| = k} \tr( \r_S^2) =   D^{-k}  \sum_{j=0}^k \binom{n-j}{n-k}  A_j(\ket{\psi}) < \binom{n}{k} \,,
  \end{equation}  
the state \(\ket{\psi}\) has to be entangled across some partition of~\(k\) versus \((n-k)\) 
parties~\footnote{See the proof of Proposition~\ref{prop:cut_relation} for details on how the combinatorial terms arise.}.

Such type of conditions readily serve as entanglement criteria in pure states, 
and were already explored in Refs.~\cite{Aschauer2004, PhysRevA.91.042339, PhysRevA.94.042302}.
Further entanglement detection methods using weight distributions and related concepts can be found in Refs.~\cite{PhysRevA.87.034301, 1710.02473v2}.

%%%% END PROPOSITION

%%%% RELATIONS
%%%%%%%%%%%%%%%%%%%%%%%%%%%%%%%%%%%

\section{Some constraints on the weight distribution}\label{sect:weight_distr_constraints}

In the following, we derive further relations on the weight distribution of pure states.
These are obtained from the Schmidt decomposition along bipartitions having fixes sizes 
and from monogamy-like relations.
As we are dealing with Ulam type problems, not the complete weight distribution is given. 
Thus, let us define the {\em reduced} weight distribution, 
which is proportional to the average distribution 
that marginals of size \(m\) of a given quantum state \(\r\) show.
This notion is useful for the Ulam type problems that we consider in this article, 
as the reduced weight distribution \(A_j^m(\r)\) can already be obtained 
from a complete set of unlabeled marginals of size \(m\). 
We refer to Sec.~\ref{sec:detect_illegitimate_decks} for details on how this can be achieved.

\begin{definition}\label{def:red_weight_distr}
Let \(\r\) be a quantum state on \(n\) parties. Given its weight distribution \(A_j(\r)\), 
 define its associated reduced weight distribution \(A_j^m(\r)\) for \(0 \leq j\leq m\) as
  \begin{equation}\label{eq:red_weight_distr}
	  A_j^m(\r) = \binom{n-j}{n-m} A_j(\r) \,.
  \end{equation}
For the following proofs we also need the weight distribution on some subsystem~\(S\subseteq \{1\dots n\}\),
\begin{equation}
  A_j^S(\r) = \!\!\!\!\sum_{\substack{ P \in P , \, \wt(P) = j \\ \supp(P) \subseteq S }} \!\!\!\! \tr(P \r) \tr(P^\dag \r)\,.
\end{equation}
\end{definition}

We now state a first constraint on the weight distribution that arises from the Schmidt decomposition.
%
%%%% PROPOSITION
\begin{proposition}[Cut-Relations]\label{prop:cut_relation}
Let \(\ket{\psi}\) be a pure state of \(n\) qubits. For all \(1\leq m \leq n\), 
  the reduced weight distributions \(A_j^m(\ket{\psi})\) satisfy
  \begin{equation}
	2^{-m} \sum_{j=0}^m A_j^m(\ket{\psi})
	= 2^{-(n-m)} \sum_{j=0}^{n-m} A_j^{n-m}(\ket{\psi}) \,.
  \end{equation}
\end{proposition}

%%%% PROOF
\begin{proof}
In the following, let us write \(A_j\) for \(A_j(\ket{\psi})\).
From the Schmidt decomposition of pure states, it follows that the purities of
reductions on complementary subsystems must be equal,
\begin{equation}\label{eq:Schmidt}
	\tr( \r_S^2) = \tr(\r_{S^c}^2)\,.
\end{equation}
Summing Eq.~\eqref{eq:Schmidt} over all bipartitions \(S\) and \(S^c\) having fixed sizes \(m \leq \nhf\) and \((n-m)\) respectively, one obtains
\begin{equation}\label{eq:bipart_sum}
	2^{-m} \sum_{|S| = m} \sum_{j=0}^m A_j^S  =  2^{-(n-m)} \!\!\sum_{|S^c| = n-m} \sum_{j=0}^{n-m} A_j^{S^c} \,.
\end{equation}
In the case of graph states, \(A_j^S\) is just the number of stabilizer elements of weight \(j\) having support in \(S\). 
Note that in Eq.~\eqref{eq:bipart_sum}, the dimensional prefactor results from the 
difference in normalization of \(\r_S\) and \(\r_{S^c}\).
By summing over all subsystem pairs of fixed size, elements of weight \(j\)
are overcounted by factors of \(\binom{n}{m}\binom{m}{j} \binom{n}{j}^{-1} = \binom{n-j}{n-m}\) and
\(\binom{n}{n-m}\binom{n-m}{j}\binom{n}{j}^{-1} = \binom{n-j}{m}\) respectively. 
We arrive at
\begin{equation}
	2^{-m} \sum_{j=0}^m \binom{n-j}{n-m} A_j 
	= 2^{-(n-m)} \sum_{j=0}^{n-m} \binom{n-j}{m} A_j \,.
\end{equation}
In terms of the reduced weight distribution, this simply reads as
\begin{equation}
	2^{-m} \sum_{j=0}^m A_j^m
	= 2^{-(n-m)} \sum_{j=0}^{n-m} A_j^{n-m} \,.
\end{equation}
This proves the claim.
\end{proof}
%%%% END PROOF
%
These are \(\lfloor (n-1)/2 \rfloor\) independent linear equations that
the weight distributions of pure states have to satisfy. Note that one of these only involves the weights up to \(A_{\nhf+1}\), which in turn can be obtained from the marginals of size \(\nhf + 1\). 
The relations can be seen as an alternate formulation of the so-called {\em quantum 
MacWilliams identity} for quantum error correcting codes in the special case of pure states~\cite{681316, 1751-8121-51-17-175301},
and generalize to states of higher local dimensions.

We now obtain further constraints on the reduced weight distributions 
that are obtained from the so-called {\em universal state inversion} and generalizations 
thereof~\cite{PhysRevLett.114.140402, 1751-8121-51-17-175301, PhysRevA.64.042315, PhysRevA.72.022311}.
In the case of qubits it can simply be attained through a spin-flip 
where every Pauli matrix in the Bloch decomposition changes sign, mapping
$I \mapsto I$,  $Y \mapsto -Y$,
$X \mapsto -X$, and $Z \mapsto -Z$~\footnote{
It is not hard to convince oneself that this spin-flip can be obtained by 
\(\tilde\r = Y^{\ot n} \r^T Y^{\ot n}\), where \(\r^T\) is the transpose of the given state \(\r\).}.
Thus when expanding a state \(\r\) in the Bloch representation [Eq.~\eqref{eq:bloch_represenation}],
the spin-flipped state shows a sign flip for all odd-body correlations,
\begin{equation}
\tilde\r 	= \frac{1}{2^n}\sum_{j=0}^n (-1)^j \!\!\!\!  \sum_{\substack{P \in \P \\ \wt(P) = j}}  \!\! \tr(P^\dag \r) P \label{eq:state_inv}\,.
\end{equation}

Note that because \(\tilde\r\) is positive semi-definite, the expression \(\tr(\r\tilde\r)\) must necessarily be non-negative. 
This leads to our next proposition.

%%%% PROPOSITION
\begin{proposition}\label{prop:shadow_on_reductions_S0}
Let \(\r\) be a state of \(n\) qubits. For all \(1 \leq m \leq n\), the reduced weight distributions \(A_j^m(\r)\) satisfy
the inequality
\begin{equation}\label{eq:shadow_on_reductions_S0}
  \sum_{j=0}^m (-1)^j A_j^m( \r) \geq 0 \,.
\end{equation}	
\end{proposition}

%%%% PROOF
\begin{proof}
We evaluate \(\tr(\r \tilde \r) \geq 0\) in the Bloch decomposition.
\begin{align}
	\tr(\r \tilde \r) &= \frac{1}{2^{2n}} \tr\Big[
				  \big(\sum_{j=0}^n (-1)^j \!\!\!\! \sum_{\substack{P \in \P \\ \wt(P) = j}}  \!\!
				  \tr(P      \r) P^\dag \big)
				  \big(\sum_{j'=0}^n  \!\! \sum_{\substack{P' \in \P \\ \wt(P) = j'}} \!\!
				       \tr(P'^\dag \r) P' \big)
				   \Big] \nn\\
			      &= \frac{1}{2^{2n}} 
				  \sum_{j=0}^n (-1)^j 	     \!\!\!\! \sum_{\substack{P \in \P \\ \wt(P) = j}} \!\!
				  \tr(P      \r) \tr(P^\dag \r) \tr(P^\dag P)  \nn\\
			      &= \frac{1}{2^n} \sum_{j=0}^n (-1)^j A_j \,\geq 0 \,.
\end{align}
Applying the same method to all reductions \(\r_S\) of fixed size \(|S| = m\),
one obtains
\begin{equation}
      \sum_{|S|=m} \tr[ \r_S \tilde\r_S] = 2^{-m} \sum_{|S|=m} \sum_{j=0}^m (-1)^j A_j^S
					     = 2^{-m} \sum_{j=0}^m (-1)^j \binom{n-j}{n-m} A_j \geq 0\,.
\end{equation}
Up to a dimensional constant, this can be rewritten as
\(  \sum_{j=0}^m (-1)^j A_j^m \geq 0 \,. \)
This ends the proof.
\end{proof}
%%%% END PROOF

For pure states, the expression \(\tr(\r \tilde\r)\) is an entanglement monotone called $n$-concurrence~\cite{PhysRevA.63.044301}.
In light of Refs.~\cite{817508, 1751-8121-51-17-175301} on the shadow enumerator of quantum codes, Eq.~\eqref{eq:shadow_on_reductions_S0}
can also be restated as the requirement that the zeroth {\em shadow coefficient} \(S_0(\r_S) = \tr(\r_S \tilde{\r}_S)\) be 
non-negative when averaged over all \(m\)-body marginals \(\r_S\).
In the case of graph states and stabilizer codes, this expression must necessarily be integer, 
as it is obtained by counting elements of the stabilizer with integer prefactors. 
Let us point to the most general form of these inequalities, the so-called {\em shadow inequality}:
\begin{theorem}[Rains~\cite{681316, 1751-8121-51-17-175301}]\label{thm:shadow_ineq}
 Let \(M\) and \(N\) be non-negative operators on \(n\) qubits. For any subset \(T \subseteq \{1\dots n\}\) it holds that
\begin{equation}\label{eq:shadow_ineq}
 \sum_{S\subseteq \{1\dots n\}} (-1)^{|S \cap T|} \tr[ \tr_{S^c}(M) \tr_{S^c}(N)] \geq 0\,,
\end{equation}
where the sum is taken over all subsets \(S\) in \(\{1\dots n\}\).
 \end{theorem}
For \(M=N=\r\), the shadow inequality represents consistency conditions for quantum states in terms of purities of reductions, and for \(T = \{1\dots n\}\) it simply corresponds to \(\tr(\r\tilde\r)\geq 0\).
By symmetrizing the shadow inequality over all subsets \(T^c\) of some fixed size $0 \leq j \leq n$
one obtains the following constraints on the weight distribution of \(n\)-qubit states~\footnote{By convention, the shadow inequality is summed over the complement \(T^c\) of \(T\), such that \(S_0(\r) = \tr(\r \tilde \r)\).}~\cite{681316, 681315, 1751-8121-51-17-175301},
\begin{equation}\label{eq:shadow_Krawtchouk}
  S_j(\r) = \sum_{0\leq k \leq n} (-1)^k K_j(k,n) A_k(\r)  \geq 0 \,.
\end{equation}
The Krawtchouk polynomial above is given by
 \begin{equation}
  K_j(k,n) = \sum_{0\leq \alpha \leq j} (-1)^\alpha  3^{j-\alpha}  \binom{n-k}{j-\alpha} \binom{k}{\alpha} \,.
 \end{equation} 
Naturally, Theorem~\ref{thm:shadow_ineq} must also hold for all reductions of a joint state \(\r\), 
these being quantum states themselves.
Demanding this condition for all reductions \(\r_S\) of a fixed size \(m\), we obtain following proposition for 
the reduced weight distribution:
\begin{proposition}\label{prop:shadow_on_reductions}
Let \(\r\) be a state of \(n\) qubits. For all \(1 \leq m \leq n\), the reduced weight distributions \(A_k^m(\r)\) must satisfy
\begin{equation}
  S_j^m(\r) = \sum_{0\leq k \leq m} (-1)^k K_j(k,m) A_k^m(\r) \geq 0 \,.
\end{equation}
\end{proposition}
\begin{proof}
Consider Eq.~\eqref{eq:shadow_Krawtchouk} for a \(m\)-body reduction of an \(n\)-qubit state. 
Summing over all marginals of size \(m\), one obtains
\begin{align}
 \sum_{\substack{S \subseteq \{1\dots n\} \\ \wt(S) = m}}  \sum_{0\leq k \leq m} (-1)^k K_j(k,m) A_k^S(\r_S) 
 &= \sum_{0\leq k \leq m} (-1)^k K_j(k,m) \binom{n-j}{n-m} A_k(\r) \nn\\
 &= \sum_{0\leq k \leq m} (-1)^k K_j(k,m) A_k^m(\r) \geq 0 \,.   
 \end{align} 
This ends the proof.
\end{proof}

%%% Connections to linear entropies
Finally, let us note that constraints on weight distributions such as Eq.~\eqref{eq:shadow_ineq} can also be expressed in terms of {\em purities} or {\em linear entropies} of reductions, and vice versa. 
The linear entropy approximates the von Neumann entropy to its first order, and is defined as
\(  S_L(\r_S) = 2[1 - \tr(\r_S^2)]\). The quantity \(\tr(\r_S^2)\) is called the purity, measuring the pureness of a state.
As an example, let us demonstrate how the universal state inversion imposes constraints 
on the linear entropies of the two- and one-party reductions of a joint state.
\begin{corollary}\label{prop:one_two_body_linent}
Let \(\r\) be a multipartite quantum state, 
and denote by \(\r_i\) and \(\r_{ij}\) its one- and two-body reductions. 
The following inequality holds,
\begin{equation}\label{eq:one_two_body_linent}
 (n-1) \sum_i S_L(\r_i) - \sum_{i<j} S_L(\r_{ij}) \geq  0\,.
\end{equation}
\end{corollary}
\begin{proof}
 As shown in Ref.~\cite{PhysRevA.72.022311}, the universal state inversion can also be written as
\begin{equation}
    \tilde\r = \sum_{S\subseteq \{1\dots n\}} (-1)^{|S|} \r_S \ot \one_{S^c}\,.
\end{equation} 
 Considering the state inversion on two-body marginals, one obtains
 \begin{align}
    \sum_{i<j} \tr[ \r_{ij} \tilde \r_{ij}] &= \sum_{i<j}\tr[\r_{ij} (\one - \r_i \ot \one_j - \one_i \ot \r_j + \r_{ij})] \nn\\
				&=  \sum_{i<j}  (1- \tr[\r_i^2] - \tr[\r_j^2]  + \tr[\r_{ij}^2]) \nn\\
				&= (n-1) \sum_i S_L(\r_i) -  \sum_{i<j} S_L(\r_{ij}) \geq 0\,.
\end{align} 
This ends the proof.
\end{proof}

Let us derive another relation arising from the shadow inequality [Eq.~\eqref{eq:shadow_ineq}] that involves three-body reductions.
\begin{corollary}\label{prop:one_two_three_body_linent}
Let \(\r\) be a multipartite quantum state. 
Denote by \(\r_i\), \(\r_{ij}\), and \(\r_{ijk}\) its one-, two-, and three-body reductions. 
The following inequality holds,
\begin{equation}\label{eq:one_two_three_body_linent}
  \sum_i S_L(\r_{i}) + \sum_{i<j} S_L(\r_{ij}) - \sum_{i<j<k} S_L(\r_{ijk})\geq 0\,.
\end{equation}
\end{corollary} 
\begin{proof}
 Consider the shadow inequality [Eq.~\eqref{eq:shadow_ineq}] on a single three-body reduction \(\r_{ABC}\). Choosing \(T=\{AB\}\) and \(M=N=\r_{ABC}\), one obtains
 \begin{equation}
  1 - \tr(\r_A^2) - \tr(\r_B^2) + \tr(\r_C^2) + \tr(\r_{AB}^2) - \tr(\r_{AC}^2) - \tr(\r_{BC}^2) + \tr(\r_{ABC}^2) \geq 0\,.
 \end{equation}
 This can be rewritten in terms of linear entropies
 \begin{equation}
  S_L(\r_A) + S_L(\r_B)  - S_L(\r_C) - S_L(\r_{AB}) + S_L(\r_{AC}) + S_L(\r_{BC}) - S_L(\r_{ABC}) \geq 0\,.
 \end{equation}
 Summing this inequality over all three-body reductions of a multipartite quantum state~\(\r\) yields the claim.
\end{proof}

Note that Corollary~\ref{prop:one_two_body_linent} simply corresponds to Propositon~\ref{prop:shadow_on_reductions_S0}
for \(m=2\) and \(j=0\), expressed in terms of linear entropies, 
and Corollary~\ref{prop:one_two_three_body_linent} corresponds to the case of \(m=3\) and \(j=1\). 
Further relations can be obtained for other values of  \(m\) and \(j\), and in turn, 
these give consistency equations on decks of quantum marginals. 
Lastly, let us point out that the shadow inequality [Eq.~\eqref{eq:shadow_ineq}] also holds as an operator inequality for all 
multipartite systems of finite local dimensions~\cite{Huber_Thesis, Eltschka2018}: for any \(T\subseteq \{1\dots n\}\), the following expression 
is positive semidefinite,
\begin{equation}
 \sum_{S\subseteq \{1\dots n\}} (-1)^{|S\cap T|} \r_S \ot \one_{S^c} \geq 0\,.
\end{equation} 
With \(T = \{1\dots n\}\) the expression reduces to that of the universal state inversion [Eq.~\eqref{eq:state_inv}].

%%%% ILLEGITIMATE DECKS
%%%%%%%%%%%%%%%%%%%%%%%%%%%%%%%%%%%

\section{Quantum deck legitimacy and state reconstruction}\label{sec:detect_illegitimate_decks}

In the spirit of Kelly's condition [Proposition~\ref{Kelly}], 
constraints imposed upon the weight distributions of quantum states can help to detect the illegitimacy of quantum decks: 
in case the (reduced) weight distribution of a putative deck does not meet the constraints given by Proposition~\ref{prop:cut_relation}, the deck can not be legitimate~\footnote{
If each card is given as a density matrix, Propositions~\ref{prop:shadow_on_reductions_S0} and~\ref{prop:shadow_on_reductions} are already fulfilled. 
Using a linear program, it can be checked however if additional weights \(A_{m+1}, \dots, A_n \geq 0\) can be found such that Eq.~\eqref{eq:shadow_Krawtchouk} holds. If not, the deck has to be illegitimate.}.
To detect the incompatibility of a deck with a joint {\em graph} state, Proposition~\ref{prop:type_I_II} can additionally be tested.

Let us first show how the weights \(A_1, \dots, A_m\) of a putative joint state can be obtained from 
having access to all \(m\)-body marginals. Given a complete quantum \(m\)-deck \(\DD_\r = \{ \r_S \,|\,\, |S| = m\}\), 
we calculate
  \begin{equation}
	\sum_{\r_S \in \DD} \sum_{\substack{P \in \P \\ \wt(P) = j}} \tr(\r_S P) \tr(\r_S P^\dag) 
	\,= \! \sum_{S,\, |S| = m}\! A_j^S 
	\,=\, \sum_{j=0}^m \binom{n-j}{n-m} A_j  \,=\, \sum_{j=0}^m A_j^m \,.
\end{equation}
From \(A_j^m\), the weights \(A_j\) can be obtained for \(0\leq j \leq m\) from Eq.~\eqref{eq:red_weight_distr}.
Note that for decks of putative joint graph states, 
\(A_j^m\) simply equals the number of stabilizer 
elements of weight \(j\) appearing in the quantum \(m\)-deck.

We now provide some examples of detecting illegitimacy using the cut-relations, 
applicable to detect the incompatibility with joint states that are pure.

%%%% EXAMPLE
\begin{example}
Consider the case of a pure three qubit state. Setting \(a=1\) in Proposition~\ref{prop:cut_relation} yields the condition \(A_2 = 3\). 
From the normalization of the state, \(\tr(\r) = 1\), it follows that \(A_1 + A_3 = 4\). 
Thus it is not possible to join three Bell states together, as each one has the weights \(A = [1,0,3]\) already.
\end{example}

%%%% EXAMPLE
\begin{example}
Let us consider a more elaborate example, the ring-cluster state of five qubits which is depicted in Fig.~\ref{fig:ulam_example_reductions}.
Those of its three-body marginals that can be obtained by tracing out nearest neighbors are an equal mixture of the four graph states that are shown in Fig.~\ref{fig:ulam_example_reductions}, where the circles denote local \(Z\)-gates.
Modifying the reductions to be the equal mixture of the states shown in the bottom row, it can be seen that no compatible joint state exists.
This follows from their corresponding weight distribution: the ring-cluster state has \(\binom{5}{3} = 10\) reductions on three qubits 
with \(A=[1,0,0,1]\). This is consistent with the cut-relations of Proposition~\ref{prop:cut_relation}, which read
\begin{align}\label{eq:five_qubit_cut_relation}
-  2 A_1  + A_2  +  A_3 		&= 10  \,\\
-  4 A_1  +  3 A_2  +  2A_3  +  A_4 	&= 35  \,.
\end{align}
Slightly modifying some reductions to be an equal mixture of the four other states that are depicted in the lower row
of Fig.~\ref{fig:ulam_example_reductions}, we obtain an illegitimate deck: these reductions have the weight distribution
\(A=[1, 0, 3/8, 11/8]\), and together with the rest of the deck, they do not satisfy Eq.~\eqref{eq:five_qubit_cut_relation}. 
Thus no compatible pure joint state on five qubits exists.
\end{example}

%%%% EXAMPLE: GHZ state
\begin{example}
Let us ask for what value \(p\in [0,1]\) a pure state \(\ket{\psi}\) on ten qubits could possibly exist, whose all reductions on six qubits equal 
\begin{equation}\label{eq:GHZ_6}
	(1-p) \frac{\one}{2^6}+ p \dyad{{\GHZ}_6}\,.
\end{equation}
Above, the Greenberger-Horner-Zeilinger state on six qubits is defined as 
\(\ket{{\GHZ}_6} = (\ket{000000} + \ket{111111} / \sqrt{2}\).
Its weights are \(A = [1,0,15,0,15,0,33]\). From it, we can obtain parts of the weight distribution of the putative joint state, 
namely
\begin{equation}
 A_{j\leq 6}(\ket{\psi}) = \binom{10}{j} \binom{6}{j}^{-1} A_j(\ket{{\GHZ}_6}) \,.
\end{equation}
Thus the putative pure joint state must have \(A = [1,0, 45p, 0, 210p, 0, 6930p, \dots]\).
Let us now see what value~\(p\) should have to satisfy Proposition~\ref{prop:cut_relation}.
The cut-relation that solely involves the weights up to~\(A_6\) requires that
\begin{equation}
   -  210 A_1  -  42 A_2  +  7 A_3  +  11 A_4  +  5 A_5  +  A_6 = 630 \,. 
\end{equation}
This can only be fulfilled if \(p = 3/35\). Linear programming however shows that no compatible weights \(A_7, \dots, A_{10}\geq 0\) satisfying Eq.~\eqref{eq:shadow_Krawtchouk} can be found. 
Thus a pure ten-qubit state having marginals as specified in Eq.~\eqref{eq:GHZ_6} can not exist.
\end{example}

\begin{figure}[tbp]\label{fig:ulam_example_reductions}
\centering
      \includegraphics[height = 7em]{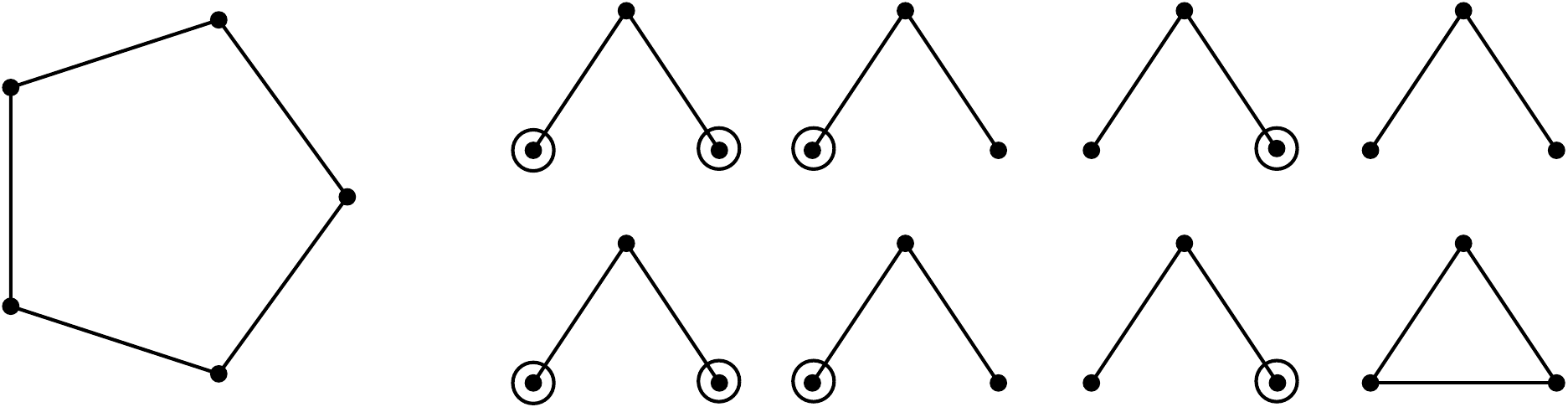}
      \caption{Left: the ring-cluster state on five qubits. 
	      Right, top row:    the three-qubit reductions of the five qubit ring-cluster state that are obtained 
	      by tracing out nearest neighbors are the equal mixture of these graph states.
	      Right, bottom row: modifying some reductions to be the equal mixture of the graph states shown in the bottom row,
	      no compatible joint state on five qubits exists.}
\end{figure}

Let us stress that in above examples, one does not require to know the particle labels.
Nonetheless it is already possible to detect some illegitimate decks when provided by a deck whose cards are of size \(\nhf+1\) only.

However, even when given the complete weight distribution \(A_0, \dots, A_n\) 
it is not always possible to decide whether or not it can indeed be realized by a quantum state: 
the constraints derived in the previous sections are necessary but not sufficient criteria for the existence of a realization.
It is not hard to find weight distributions that satisfy all known constraints,
but for which no corresponding quantum state can exist.
As an example, consider a hypothetical pure state of seven qubits
that shows maximal entanglement across every bipartition, therefore having all three-body marginals maximally mixed~\footnote{
This a so-called absolutely maximally entangled state, having the code parameters \(((7,1,4))_2\).}. 
Its weight distribution reads \( A = [1,   0,   0,   0,  35,  42,  28,  22]\)~\cite{PhysRevA.69.052330}.
While it was known by exhaustive search that the distribution cannot be realized by a graph state (in other words, cannot be graphical), 
it was only recently shown that no pure state with such property can exist at all~\cite{PhysRevLett.118.200502}.

Further cases of weight distributions are known whose realizations are still unresolved.
As an interesting example, the existence of a graph state on \(24\) qubits, 
having all \(9\)-body reductions maximally mixed, is a long-standing open problem~\footnote{This state is equivalent to a self-dual additive code over \(\operatorname{GF}_4\), and corresponds to a quantum code having the parameters \([[24,0,10]]_2\). 
See also Research Problem \(13.3.7\) in Ref.~\cite{Nebe2006}
and the code tables of Ref.~\cite{Grassl:codetables}.}.
Putative weights for such a state of type \(II\), having even weights only, are~\footnote{This weight distribution can also be found in 
the On-Line Encyclopedia of Integer Sequences, see \url{http://oeis.org/A030331}.}
\begin{align}
&[A_{10} , A_{12}, A_{14}, \dots A_{24}] = \nn\\
 &[   18216,   156492,  1147608,  3736557,  6248088,  4399164,  1038312,
    32778] \,.
\end{align}

It is worth to note that graph states are not uniquely identified by their weight distributions; 
graph states inequivalent under LU-transformations and graph isomorphism can indeed have 
the same weight distribution. This is to be expected, as the weight distribution consists of polynomial invariants 
that are of degree two only~\cite{Aschauer2004, 817508}.
As an example, consider the two seven-qubit graph states that are depicted in Fig.~\ref{fig:LU_ineq_same_distr}.
These can be shown to be inequivalent under local unitaries and graph isomorphism, 
but they share the same weight distribution of
\( A = [1,   0,   0,   7,  21,  42,  42, 15]\)~\footnote{
These are the graphs No. \(42\) and \(43\) of Fig. \(5\) in 
Ref.~\cite{PhysRevA.69.062311}.}. 
We conclude that graph states are not uniquely identified by their weights.

\begin{figure}[tbp]\label{fig:LU_ineq_same_distr}
\centering
      \includegraphics[height = 6em]{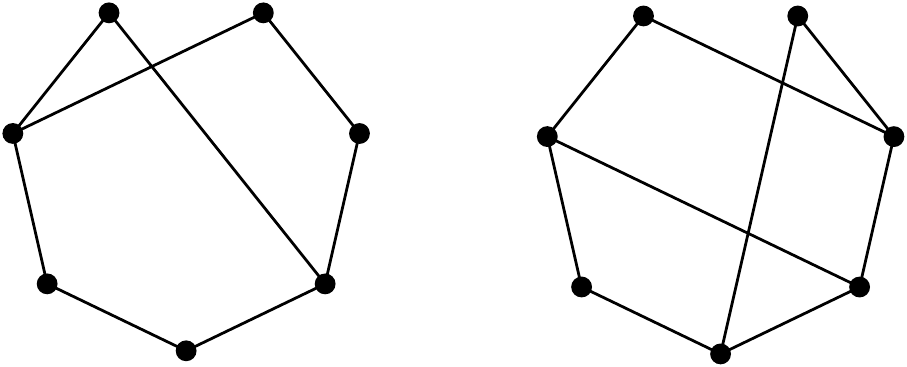}
      \caption{Two graph states on seven qubits that share the same 
      weight distribution, 
      but which can be shown to be inequivalent under local unitaries and graph isomorphism.
      These correspond to graphs No. \(42\) and \(43\) of Fig. \(5\) in Ref.~\cite{PhysRevA.69.062311}.}
\end{figure}

Let us end by providing an answer to the analogue of Ulam's reconstruction problem for general quantum states.
Can every joint state uniquely be reconstructed if all its \((n-1)\)-body marginals are known? The answer is no: 
as a counter example, consider a deck all whose cards equal the state \( ( \dyad{0}^{\otimes (n-1)} + \dyad{1}^{\otimes (n-1)})/2\). It is not hard to see that 
for all \(\theta\in [0, 2\pi]\), the Greenberger-Horne-Zeilinger type states 
\begin{equation}
 \ket{\GHZ_{\theta}} = \frac{1}{\sqrt{2}}(\ket{0}^{\otimes n} + e^{i\theta} \ket{1}^{\otimes n})
\end{equation}
are compatible. 

This leaves open the possibility that a reconstruction can nevertheless be achieved for almost all states, 
and the works of Refs.~\cite{PhysRevA.71.012324, PhysRevA.96.010102} indeed point towards that direction: 
considering generic pure states, the particle labels can be reconstructed by the comparison of one-body reductions. Ref.~\cite{PhysRevA.71.012324} has shown that generic pure states are uniquely determined by a certain set of~\(\nhf \) labeled marginals of size \(\lceil n/2 \rceil+1\). 
Alternatively, almost all pure states can (amongst pure states) uniquely be reconstructed from a carefully chosen set of three labeled marginals of size \((n-2)\) ~\cite{PhysRevA.96.010102}. Thus generic pure states are reconstructible from an incomplete deck: considering \((\lceil n/2 \rceil+1)\)-decks, almost all pure states have the reconstruction number \(\operatorname{rn}_\text{mix} (\ket{\psi})=\nhf\), while \(\operatorname{rn}_\text{pure}(\ket{\psi}) = 3\) when considering their \((n-2)\)-decks.

\section{Conclusion}\label{sect:conclusion}
We have introduced the analogue of the Ulam graph reconstruction problem to the case of quantum graph states.
In contrast to the classical case, a joint graph state can (up to local \(Z\)-gates)  
be reconstructed from a single card in the deck, and is thus reconstructible 
amongst graph states~\cite{1751-8121-48-9-095301}.
The result by Bollobás~\cite{JGT:JGT3190140102}, namely, that almost every
graph can uniquely be reconstructed from a carefully chosen set of three cards, 
has also an interesting counterpart in the quantum setting:
considering generic pure quantum states, particle labels can be restored by the comparison of one-body marginals, rendering almost all pure states reconstructible from three carefully chosen but initially unlabeled 
marginals of size \((n-2)\)~\cite{PhysRevA.96.010102}. It would be desirable to understand whether or not similar results apply to the special case of graph states.

As in the classical setting, the legitimate deck problem is of interest. 
While finding a complete solution to this problem is likely to be computationally hard, 
consistency relations on the weight distributions allow the detection of 
some but not all illegitimate quantum decks; in some cases this is already possible 
when the marginals are of size \(\nhf+1\). 
It would be interesting to see whether similar relations can be obtained for classical decks of graphs.

\bigskip

\noindent \emph{Acknowledgements.} We thank Danial Dervovic, Chris Godsil, Otfried Gühne, Joshua Lockhart, David W. Lyons, Christian Schilling, and Jens Siewert for fruitful discussions.
We acknowledge support by 
the Swiss National Science Foundation (Doc.Mobility grant 165024), 
the Fundación Cellex, 
the ERC (Consolidator Grant 683107/TempoQ),
the National Natural Science Foundation of China, 
The Royal Society, 
the UK Engineering and Physical Sciences Research Council (EPSRC), 
and the UK Defence Science and Technology Laboratory (Dstl). 
This is part of the collaboration between US DOD, UK MOD and UK EPSRC under the Multidisciplinary University Research Initiative.

\bibliographystyle{apsrev4-1}
\bibliography{ulam}
\end{document}